\documentclass{article}
\usepackage{graphicx,amssymb,amsmath}
\usepackage{tda}
\title{Learning Simplicial Complexes from Persistence Diagrams}

\author{Robin Lynne Belton\thanks{Depart. of Mathematical Sciences,
    Montana State U.}
    \and
        Brittany Terese Fasy\footnotemark[1]~\footnotemark[2]
        \and
        Rostik Mertz\thanks{School of Computing,
            Montana State
            U. \newline
   {\scriptsize {\tt \{robin.belton, brittany.fasy, david.millman, anna.schenfisch, jordan.schupbach\}@montana.edu}
    {\tt \{samuel.micka, daniel.salinas, lucia.williams\}@msu.montana.edu}
    {\tt rostik.mertz@student.montana.edu}}}
        \and
        Samuel Micka\footnotemark[2]
        \and
        David L. Millman\footnotemark[2]
        \and
        Daniel Salinas\footnotemark[2]
        \and
        Anna Schenfisch\footnotemark[1]
        \and
        Jordan Schupbach\footnotemark[1]
        \and
        Lucia Williams\footnotemark[2]
    }

\index{Belton, Robin Lynne}
\index{Fasy, Brittany Terese}
\index{Mertz, Rostik}
\index{Micka, Samuel}
\index{Millman, David L.}
\index{Salinas, Daniel}
\index{Schenfisch, Anna}
\index{Schupbach, Jordan}
\index{Williams, Lucy}

\begin{document}

\thispagestyle{empty}
\maketitle

\begin{abstract}
Topological Data Analysis (TDA) studies the ``shape'' of data. A common
topological descriptor is the persistence diagram, which encodes topological
features in a topological space at different scales. Turner, Mukherjee, and
Boyer showed that one can reconstruct a simplicial complex embedded in $\R^{3}$
using persistence diagrams generated from all possible height filtrations
(an uncountably infinite number of directions).  In this paper, we present an
algorithm for reconstructing plane graphs~$\simComp = (V,E)$ in $\R^2$, i.e., a
planar graph with vertices in general position and a straight-line embedding,
from a quadratic number height filtrations and their respective
persistence diagrams.
\end{abstract}

\section{Introduction}
\label{sec:intro}

Topological data analysis (TDA) is a promising tool in fields as varied as
materials science, transcriptomics, and neuroscience \cite{giusti2015clique,
lee2017quantifying, rizvi2017single }. Although TDA has been quite successful in
the analysis of point cloud data \cite{nicolau2011topology}, its purview extends
to any data that can be encoded as a topological space. Topological spaces can
be described in terms of their homology, e.g., connected components and ``holes.''
Simplicial complexes, in particular, are the most common representation of
topological spaces.
In this work, we focus our attention on a subset of simplicial complexes,
namely, plane graphs embedded in $\R^2$, with applications to shape
reconstruction.

In this paper, we explore the question, \emph{Can we reconstruct embedded
simplicial complexes from a finite number of directional persistence diagrams?}
Our work is motivated by \cite{turner2014persistent}, which proves that one can
reconstruct
simplicial complexes from an uncountably infinite number of
diagrams. Here, we make the first step towards providing a polynomial-time
reconstruction for simplicial complexes.
In particular, the main contributions of this paper are to set a bound on the number
of persistence diagrams required to reconstruct a plane graph
and to provide a polynomial-time algorithm for reconstructing the~graph.

\section{Related Work}

The problem of manifold and stratified space learning is an active research area
in computational mathematics. For example, Zheng et al.\
study the 3D reconstruction of plant roots from multiple 2D
images~\cite{rootreconstruction}. Their method uses persistent homology
to ensure the resulting 3D root model is connected.

Map construction algorithms reconstruct street maps as an embedded graph from a
set of input trajectories.  Three common approaches are Point Clustering,
Incremental Track Insertion, and Intersection Linking~\cite{maps}. Ge, Safa,
Belkin, and Wang develop a point clustering algorithm using Reeb graphs to
extract the skeleton graph of a road from point-cloud data~\cite{ge2011data}.
The original embedding can be reconstructed using a principal curve
algorithm~\cite{kegl2000learning}.  Karagiorgou and Pfoser give an incremental
track insertion algorithm to reconstruct a road network from vehicle trajectory
GPS data~\cite{ili}. Ahmed et  al.\  provide an incremental track insertion
algorithm to reconstruct road networks from point could data~\cite{iti}. The
reconstruction is done incrementally, using a variant of the Fr\'echet distance
to add curves to the current basis.  Ahmed, Karagiorgou, Pfoser, and Wenk
describe all these methods in~\cite{maps}. Finally, Dey, Wang, and Wang use
persistent homology to reconstruct embedded graphs. This research has also been
applied to input trajectory data~\cite{dey2018graph}. Dey et al.\ use
persistence to guide the Morse cancellation of critical simplices. In contrast,
the work presented here uses persistence to generate the diagrams that encode
the underlying graph.

Our work extends previous work on the persistent homology transform
(PHT)~\cite{turner2014persistent}. As detailed in \secref{preliminary},
persistent homology summarizes the homological changes for a filtered
topological space. When applied
to a simplicial complex embedded in $\R^d$, we can compute a different
filtration for every direction in $\sph^{d-1}$; this family of persistence
diagrams is referred to
as the persistent homology transform (PHT).  The map from a simplicial complex
to PHT is injective
\cite{turner2014persistent}.  Hence, knowing the PHT of a simplicial complex uniquely
identifies that complex. The proof presented in~\cite{turner2014persistent}
relies on the continuity of persistence diagrams as the direction of
filtration varies \emph{continuously}.

Our paper
bounds the number of directions by presenting an algorithm for reconstructing
the simplicial complex, when we are able to obtain persistence diagrams for a
given set of directions.  Simultaneous to our investigation, others have also
observed that the number of directions can be bounded using the
Radon transform;
see~\cite{ghrist2018euler,curry2018directions}.
In the work presented in the current  paper,
we seek to reconstruct graphs from their respective
persistence diagrams, using a geometric approach. We bound the number of directional
persistence diagrams since computing the PHT, as presented in~\cite{turner2014persistent},  requires the
computation of filtrations from an infinite number of possible directions. Our
work provides a theoretical guarantee of correctness for a finite subset of
directions by providing the reconstruction algorithm.

\section{Preliminary}
\label{sec:preliminary}
In this paper, we explore the question,
\emph{Can we reconstruct embedded simplicial complexes from a finite number of
directional persistence diagrams?} We begin by summarizing the necessary
background information, but refer the reader  to
\cite{edelsbrunner2010computational} for a more comprehensive~overview of
computational topology.

\paragraph{Simplices and Simplicial Complexes}
Intuitively, a $k$-simplex is a $k$-dimensional generalization of a triangle,
i.e., a zero-simplex is a vertex, a one-simplex is an edge connecting two
vertices, a two-simplex is a triangle, etc. In this paper, we focus on a subset
of simplicial complexes embedded in $\R ^2$ consisting of only vertices and
edges. Specifically, we study \emph{plane graphs} with straight-line embeddings
(referred to simply as \emph{plane graphs} throughout this paper). Furthermore,
we assume that the embedded vertices are in general position, meaning that no
three vertices are collinear and no two vertices share an $x$- or
$y$-coordinate.

\paragraph{Height Filtration}
Let $\simComp$ be a plane graph and denote $\sph ^1$ as the unit sphere in
$\R^2$. Consider $\dir \in \sph^1$; we define the \emph{lower star filtration}
with respect to direction $\dir$ in two steps. First, we let
$\hFiltFun{\dir}: \simComp\rightarrow \R$ be defined for a simplex
$\sigma \subseteq \simComp$ by $\hFiltFun{\dir}(\sigma) = \max_{v \in \sigma}
\dprod{v}{\dir},$ where~$\dprod{x}{y}$ is the inner (dot) product and measures
height in the direction of $y$, if $y$ is a unit vector. Intuitively, the height
of $\sigma$ from $\dir$  is the maximum height of all vertices in~$\sigma$.
Then, for each $t \in \R$, the
subcomplex~$\simComp_t:=\hFiltFunT{\dir}{[-\infty,t)}$ is composed of all
simplices that lie entirely below or at the height~$t$, with respect to the
direction $\dir$. Notice~$\simComp_r \subseteq \simComp_t$ for all~$r \leq t$
and~$\simComp_r=\simComp_t$ if no vertex has height in the interval $[r,t]$.
The sequence of all such subcomplexes, indexed by~$\R$, is the
\emph{height filtration} with respect to~$\dir$, notated as $\hFiltComplex{\dir}{\simComp}$.
Often, we simplify notation and define~$\hFilt{\dir} := \hFiltComplex{\dir}{\simComp}$.

\paragraph{Persistence Diagrams}

The persistence diagram is a summary of the homology groups $H_*(\simComp_t)$
as the height parameter $t$ ranges from $-\infty$ to $\infty$; in particular,
the persistence diagram is a set of birth-death pairs~$(b_i,d_i)$.  Each pair
represents an interval~$[b_i,d_i)$ corresponding to a homology generator. For
example, a birth event may occur when the height filtration discovers a new
vertex, representing a new component, and the corresponding death represents
the vertex joining another connected component. By definition
\cite{edelsbrunner2010computational}, all points in the diagonal~$y=x$ are also
included with infinite multiplicity. However, in this paper, we consider only
those points on the diagonal that are explicitly computed in the persistence
algorithm found in \cite{edelsbrunner2010computational}, which correspond to features with the same birth and death time.
For a direction $\dir \in \sph^1$, let the \emph{directional
persistence diagram}~$\dgm{i}{\hFiltComplex{\dir}{\simComp}}$ be the set of
birth-death pairs for the $i$-th homology group from the height
filtration~$\hFiltComplex{\dir}{\simComp}$.  As with the height filtration, we
simplify notation and define~$\dgm{i}{\dir} :=
\dgm{i}{\hFiltComplex{\dir}{\simComp}}$ when the complex is clear from context.
We conclude this section with a remark relating birth-death pairs in persistence
diagrams to the  simplices in $\simComp$; a full discussion of this remark is
found in~\cite[pp.~$120$--$121$ of \S $V.4$]{edelsbrunner2010computational}.

\begin{remark}[Adding a Simplex]
Let $\simComp$ be a simplicial complex and $\sigma$ a $k$-simplex whose
faces are all in~$\simComp$. Let $\beta_i$ refer to the $i$-th Betti number,
i.e., the rank of the $i$-th homology group. Then, the addition of $\sigma$ to
$\simComp$ will either increase $\beta_{k}$ by one or decrease $\beta_{k-1}$
by one.
\label{rem:addSimp}
\end{remark}

Thus, we can form a bijection between simplices of $\simComp$ and birth-death
events in a persistence diagram. This observation is the crux of the proofs of
 \thmref{intComp}  in \secref{vRec} and  \lemref{Indegree} in \secref{eRec}.

\section{Vertex Reconstruction}
\label{sec:vRec}

In this section, we present an algorithm for recovering the locations of
vertices of a simplicial complex~$\simComp$ using three directional persistence
diagrams.  Intuitively, for each direction, we identify the lines on which the
vertices of $\simComp$  must lie. We show that by choosing the three directions
such that they satisfy a simple property,  we can identify all vertex locations
by searching for points in the plane where three lines intersect. We call these
lines \emph{filtration~lines}:

\begin{definition}[Filtration Lines]
Given a direction vector $\dir \in \sph^1$, and a height $h \in \R$ the
\emph{filtration line at height $h$} is the line, denoted $\pLine{\dir}{h}$,
through point~$h*\dir$ and perpendicular to direction $\dir$, where $*$ denotes
scalar multiplication. Given a finite set of vertices $V \subset \R^2$, the
\emph{filtration lines of~$V$} are the set of lines
$$
    \pLines{\dir}{V} = \{ \pLine{\dir}{\hFiltFun{\dir}(v)} \}_{v \in V}.
$$
\end{definition}
Notice that all lines in $\pLines{\dir}{V}$ are parallel. Intuitively, if $v$ is a
vertex in a simplicial complex~$\simComp$, then the line
$\pLine{\dir}{\hFiltFun{\dir}(v)}$ occurs at the height where the filtration
$\hFilt{\dir}(K)$ includes $v$ for the first time. If the height is known but the
complex is not, the line~$\pLine{\dir}{\hFiltFun{\dir}(v)}$ defines all potential
locations for $v$.  By \remref{addSimp}, the births in the zero-dimensional
persistence diagram are in one-to-one correspondence with the vertices of
the simplex complex $K$.  Thus, we can construct~$\pLines{\dir}{V}$ from
a single directional diagram in $O(n)$~time. Given filtration lines for three
carefully chosen directions, we next show a correspondence between
intersections of three filtration lines and vertices in $\simComp$.

In what follows, given a direction $s_i \in \sph^1$ and a point $p \in \R^2$, define
$\simplePLine{i}{p} := \pLine{\dir_i}{\hFiltFun{\dir_i}(p)}$ as a way to simplify notation.

\begin{lemma}[Vertex Existence Lemma]
    Let~$\simComp$ be a simplicial complex with vertex set $V$ of size~$n$.
    Let $\dir_1,\dir_2 \in \sph^1$ be linearly independent and further suppose
    that $\pLines{\dir_1}{V}$ and $\pLines{\dir_2}{V}$ each
    contain~$n$ lines. Let $A$ be the collection of vertices at the intersections of
    lines in $\pLines{\dir_1}{V} \cup \pLines{\dir_2}{V}$. Let $\dir_3 \in \sph^1$ such
    that for all $u, v \in A$, $\simplePLine{3}{u} = \simplePLine{3}{v} \iff u = v$.
    Then, the following two statements hold true:

        \emph{(1)} $v \in V \iff \simplePLine{3}{v} \in \pLines{\dir_3}{V}$
            and
            $A \cap \simplePLine{3}{v} = \{v\}$

        \emph{(2)} For all $\ell \in \pLines{\dir_3}{V}$, $A \cap \ell \neq
        \emptyset$.
            \label{part:allIntersect}
\label{lem:vertExist}
\end{lemma}

\begin{figure}
\begin{center}
\includegraphics{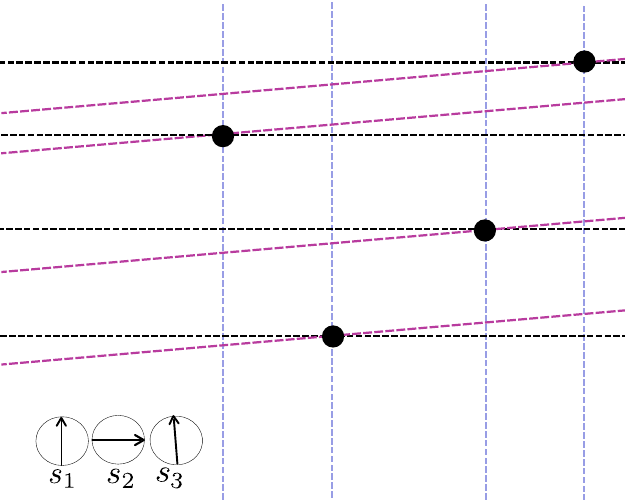}
    \caption{A vertex set, $V$, of size~$4$ with filtration lines that satisfy
the Vertex Existence Lemma. Here, $\dir_1, \dir_2 \in \sph^1$ are linearly
independent and the filtration lines are colored so that $\pLines{\dir_1}{V}$ are
the black horizontal lines, $\pLines{\dir_2}{V}$ are the blue vertical lines, and $\pLines{\dir_3}{V}$
are the magenta diagonal lines. An intersection of three colored lines
corresponds to the location of a vertex in $V$.}
\label{fig:vert}
\end{center}
\end{figure}

\begin{proof}

    First, we prove Part (1).

    ($\Rightarrow$) Let $v \in V$. Then,
    $\simplePLine{i}{v} \in \pLines{\dir_i}{V}$ and $v \in \simplePLine{i}{v}$ for
    $i = \{1, 2, 3\}$. Hence,
    $\simplePLine{1}{v} \cap \simplePLine{2}{v} \cap \simplePLine{3}{v}= \{v\}$,
    as desired.

    ($\Leftarrow$) Assume, for the sake of contradiction, that
    $\simplePLine{3}{v} \in \pLines{\dir_3}{V}$ and
    $A \cap \simplePLine{3}{v} = \{v\}$,
    yet $v \not\in V$.
    Since $\simplePLine{3}{v} \in \pLines{\dir_3}{V}$ and
    $v \not\in V$,
    some other vertex $u \in V$ must have height $\hFiltFun{\dir_3}(v)$.
    Since~$u \in V$, we know
    $\simplePLine{i}{u} \in \pLines{\dir_i}{V}$ for $i \in \{1, 2, 3\}$.
    And, by~($\Rightarrow$) applied to $u$, we know $u \in A$.
    Since $\simplePLine{3}{u} = \simplePLine{3}{v}$,
    both $u$ and~$v$ are in $A$ and on the line $\simplePLine{3}{v}$, but $u
    \neq v$, which is a~contradiction.

    Next, we prove Part (2) of the lemma. Assume, for contradiction, that
    there exists $\ell \in \pLines{\dir_3}{V}$ such that $A \cap \ell = \emptyset$.
    As $\ell \in \pLines{\dir_3}{V}$, a vertex
    $v \in V$ exists such that $\ell = \simplePLine{3}{v}$ and $v$ lies on $\ell$.
    However, $v \in \simplePLine{1}{v} \cap \simplePLine{2}{v} \subset A$,
    which is a contradiction.
\end{proof}

In the previous lemma, we needed to find a third direction with specific properties.
If we use horizontal and vertical lines for our first two directions, then we
can use the geometry of the boxes formed from these lines to pick the third
direction. More specifically, we look at the box with the largest width and
smallest height and pick the third direction so that if one of the corresponding
lines intersects the bottom left corner of the box then it will also intersect
the box somewhere along the right edge. In~\figref{vert}, the third direction
was computed using this procedure with the second box from the left in the top
row. Next, we give a more precise description of the vertex localization procedure.

\begin{lemma}[Vertex Localization]
    Let $L_H$ and $L_V$ be~$n$ horizontal and $n$ vertical lines, respectively.
    Let~$w$ (and $h$) be the largest (and smallest) distance between two lines of
    $L_V$ (and $L_H$, respectively). Let $B$ be the smallest axis-aligned
    bounding box containing the intersections of lines in $L_H \cup L_V$.
    For $0 < \varepsilon < h$, let $\dir = (w, h - \varepsilon )$.
    Any line parallel to $\dir$ can intersect at most one line of $L_H$ in $B$.

    \label{lem:buffer}
\end{lemma}

\begin{proof}
    Note that, by definition, $\dir$ is a vector in the direction that is at a slightly
    smaller angle than the diagonal of the box of size $w$ by $h$.
    Assume, by contradiction, that a line parallel to $\dir$ may intersect two
    lines of $L_H$ within~$B$. Specifically, let $\ell_1, \ell_2 \in L_H$ and
    let $\ell_{\dir}$ be a line parallel to $\dir$ such that the points
    $\ell_i \cap \ell_{\dir} = (x_i, y_i)$ for $i = \{1, 2\}$ are the two such intersection
    points within~$B$. Notice since the lines of $L_H$ are horizontal
    and by the definition of $h$, we observe that $|y_1-y_2|\geq h$.
    Let $w' = |x_1 - x_2|$, and observe $w' \leq w$.  Since the slope of $\ell_s$ is
    $(h - \varepsilon)  / w$, we have $|y_1 - y_2| < h$, which is a contradiction.
\end{proof}

We conclude this section with an algorithm to determine the coordinates of the vertices
of the original graph in~$\R^2$, using only three height filtrations.

\begin{theorem}[Vertex Reconstruction]
    Let~$\simComp$ be a plane graph. We can can compute the coordinates of all
    $n$ vertices of $\simComp$ in $O(n \log{n})$ time from three directional
    persistence diagrams.
\label{thm:intComp}
\end{theorem}

\begin{proof}
Let $\dir_1 = (1, 0)$ and $\dir_2 = (0, 1)$, which are linearly independent.
    We compute the filtration lines~$\pLines{\dir_i}{V}$ for $i=1,2$ in $O(n)$
    time by \remref{addSimp}. By our general position assumption, no two
    vertices of~$\simComp$ share an $x$- or $y$-coordinate. Thus, the
    sets~$\pLines{\dir_1}{V}$ and~$\pLines{\dir_2}{V}$ each contain~$n$
    distinct lines. Let $A$ be the set of intersection points of the lines in
    $\pLines{\dir_1}{V}$ and~$\pLines{\dir_2}{V}$. The next step is to identify
    a direction $\dir_3$ such that each line in~$\pLines{\dir_3}{V}$ intersects
    with only one point $A$, so that we can use \lemref{vertExist}.

    Let $w$ (and $h$) be the greatest (and least) distance between two adjacent
    lines in $\pLines{\dir_1}{V}$ (and $\pLines{\dir_2}{V}$, respectively). Let $B$
    be the smallest axis-aligned bounding box containing $A$, and let
    $\dir_{*} = (w, \frac{h}{2})$. By \lemref{buffer}, any line parallel to
    $\dir_{*}$ will intersect at most one line of~$\pLines{\dir_2}{V}$ in $B$.
    Thus, we choose $\dir_3 \in \sph^1$ that is perpendicular to  $\dir_{*}$.
    By the second part of~\lemref{vertExist}, we now have that each line in
    $\pLines{\dir_3}{V}$ intersects $A$. Thus, there are $n$ intersections
    between $\pLines{\dir_2}{V}$ and $\pLines{\dir_3}{V}$ in $B$, each of
    which also intersects with  $\pLines{\dir_1}{V}$.

    The previous paragraph leads us to a simple algorithm for finding the third
    direction and identifying all the triple intersections.  In the analysis below,
    steps that do not mention a number of diagrams use no diagrams.
    First, we construct $\pLines{s_1}{V}$ and $\pLines{s_2}{V}$ in $O(n)$ time
    using two directional persistence diagrams.
    Second, we sort the lines of
    $\pLines{s_1}{V}$ and $\pLines{s_2}{V}$ by their $x$- and $y$-intercepts
    respectively in $O(n \log n)$ time. Third, we find $\dir_3$ by computing
    $w$ and $h$ from our sorted lines in $O(n)$ time.
    Fourth, we construct
    $\pLines{s_3}{V}$ in $O(n)$ time using one directional persistence diagram.
    Fifth, we sort the lines in $\pLines{s_3}{V}$ by their intersection with the
    leftmost line of $\pLines{s_1}{V}$ in $O(n \log n)$ time. Finally, we compute
    coordinates of the $n$ vertices by intersecting the $i$-th line of
    $\pLines{s_2}{V}$ with the $i$-th line of $\pLines{s_3}{V}$
    in~$O(n)$ time.  (Observe, this last step works since the vertices correspond
    to the $n$ intersections in $B$, as described above).

    Therefore, we use three directional diagrams (two in the first step and one
    in the fourth step) and $O(n \log n)$ time (sorting of lines in the second and
    fifth steps) to reconstruct the vertices.
\end{proof}

\section{Edge Reconstruction}
\label{sec:eRec}
Given the vertices constructed in \secref{vRec}, we describe how to reconstruct the
edges in a plane graph using $n(n-1)$ persistence diagrams. The key to determining
whether an edge exists or not is counting the degree of a vertex, for edges ``below''
the vertex with respect to a given direction. We begin this section by defining
necessary terms, and then explicitly describing our algorithm for constructing edges.

\begin{definition}[Indegree of Vertex]
Let $\simComp$ be a plane graph with vertex set $V$. Then, for every vertex
$v \in V$ and every direction $\dir \in \sph^1$, we define:
$$\indeg{v}{\dir} = |\{(v,v') \in E \mid \dir \cdot v' \leq \dir \cdot v \}|.$$
\end{definition}
In other words, the indegree of $v$ is the number of edges incident to $v$ that lie
below $v$, with respect to direction~$\dir$; see~\figref{indegree}.
\begin{figure}[htb]
\begin{center}
\includegraphics{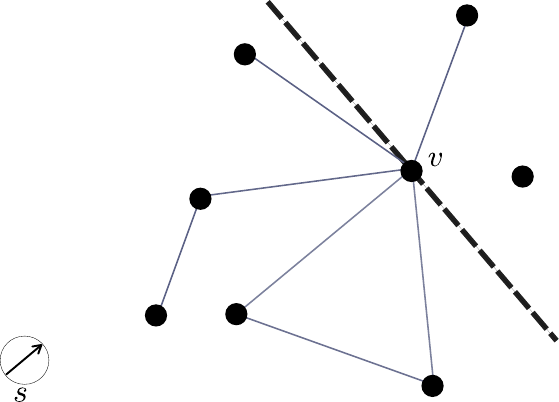}
\caption{A plane graph with a dashed line drawn intersecting $v$ in the direction
perpendicular to~$\dir$. Since four edges incident to $v$  lie below $v$, with
respect to direction~$\dir$, $\indeg{v}{\dir}=4$.}
\label{fig:indegree}
\end{center}
\end{figure}

Next, we prove that given a direction, we can determine the indegree of a
vertex:

\begin{lemma}[Indegree from Diagram]\label{lem:Indegree}
Let $\simComp$ be a plane graph with vertex set $V$. Let  $\dir \in \sph ^1$
be such that no two vertices are at the same height with respect to $\dir$, i.e.,
$|\pLines{\dir}{V}| = n$. Let $\dgm{0}{\dir}$ and~$\dgm{1}{\dir}$ be the zero-
and one-dimensional persistence diagrams resulting from the height filtration
$\hFilt{\dir}$ on $\simComp$. Then, for all $v \in V$,
\begin{equation*}
\begin{split}
\indeg{v}{\dir} = & |\{(x,y) \in \dgm{0}{\dir} \mid y = v \cdot \dir \}| +\\
& |\{(x,y) \in \dgm{1}{\dir} \mid x = v \cdot \dir \}|.
\end{split}
\end{equation*}
\end{lemma}

\begin{proof}
Let $v, v' \in V$ such that $\dir \cdot v' < \dir \cdot v$, i.e., the vertex~$v'$
is lower in direction $\dir$ than $v$. Then, by~\remref{addSimp}, if
$(v, v') \in E$, it must be one of the following in the filter of~$\simComp$
defined by $\dir$:
    (1) an edge that joins two disconnected components; or
    (2) an edge that creates a one-cycle.
Since edges are added to a filtration at the height of the higher vertex,
we see~(1) as a death in $\dgm{0}{\dir}$ and (2) as a birth in~$\dgm{1}{\dir}$,
both at height~$\dir \cdot v$. In addition, each finite death in $\dgm{0}{\dir}$
and every birth in $\dgm{1}{\dir}$ at time $\dir \cdot v$ must correspond to an
edge, i.e., edges are the only simplices that can cause these events. Then,
the set of edges of types (1) and (2) is
$\{(x,y) \in \dgm{0}{\dir} \mid y = v \cdot \dir \}$ and
$\{(x,y) \in \dgm{1}{\dir} \mid x = v \cdot \dir \}$, respectively. The size of
the union of these two multi-sets is equal to the number of edges starting at~$v'$
lower than~$v$ in direction~$s$ and ending at~$v$, as required.
\end{proof}

In order to decide whether an edge $(v,v')$ exists between two vertices, we look at the
degree of $v$ as seen by two close directions such that $v'$ is the only vertex in what
we call a \emph{bow tie at $v$}:
\begin{definition}[Bow Tie]
    Let ~$v \in V$, and choose $\dir_1,\dir_2 \in \sph^1$.  Then, a bow tie at $v$ is the
    symmetric difference between the half planes below $v$ in directions
    $\dir_1$ and~$\dir_2$. The \emph{width} of the bow tie is half of the angle
    between $\dir_1$ and $\dir_2$.
\end{definition}

Because no three vertices in our plane graph are collinear, for each
pair of vertices $v,v'\in V$, we can always find a bow tie centered at $v$ that
contains the vertex~$v'$ and no other vertex in $V$; see
\figref{bowtie}.
\begin{figure}[htb]
\begin{center}
\includegraphics{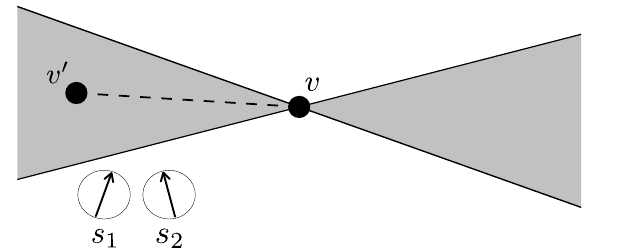}
\caption{A bow tie~$B$ at~$v$, denoted by the gray shaded area.~$B$ contains
exactly one vertex,~$v'$, so the only potential edge in~$B$ is $(v,v')$.}
\label{fig:bowtie}
\end{center}
\end{figure}
We use bow tie regions that only contain one vertex, $v'$ other than the center,
$v$ to determine if there exists an edge between $v$ and $v'$; see \figref{edgeExist}.
We then use \lemref{edgeExist} to decide if the edge $(v,v')$ exists in our plane graph.
\begin{figure}[htb]
\begin{center}
\includegraphics{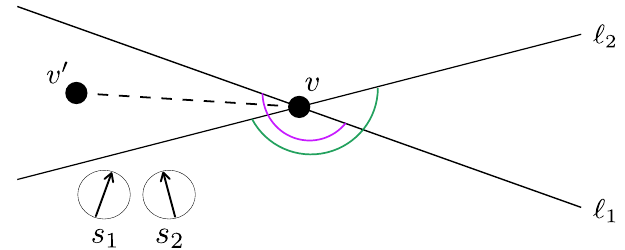}
\caption{A bow tie at $v$ that contains the vertex~$v'$ and no other vertices.
In order to determine if there exists an edge between~$v$ and~$v'$, we
compute~$\indeg{v}{\dir_1}$ and~$\indeg{v}{\dir_2}$, i.e., the number of edges
incident to~$v$ in the purple and green arcs, respectively. An edge
exists between~$v$ and~$v'$ if and only if~$|\indeg{v}{\dir_1} - \indeg{v}{\dir_2}|=1.$}
\label{fig:edgeExist}
\end{center}
\end{figure}

\begin{lemma}[Edge Existence] \label{lem:edgeExist}
Let $\simComp$ be a plane graph with vertex set $V$ and edge set $E$.
Let~\mbox{$v,v' \in V$}. Let $\dir_1, \dir_2 \in  \sph ^1$ such that the bow
tie~$B$ at $v$ defined by $\dir_1$ and $\dir_2$ satisfies:
$B \cap V = v'$.~Then,
$$|\indeg{v}{\dir_1} - \indeg{v}{\dir_2}|=1 \iff (v,v') \in E.$$
\end{lemma}
\begin{proof}
    Since edges
    in~$\simComp$ are straight lines, any edge incident to~$v$ will either fall
    in the bow tie region~$B$ or will be on the same side (above or below) of
    both lines. Let $A$ be the set of edges incident on $v$ and below both
    lines; that is \mbox{$A = \{ (v,v_*) \in E~|~ \dprod{\dir_i}{v_*} <
    \dprod{\dir_i}{v} \}.$} Furthermore, suppose we split the bowtie into the
    two infinite cones. Let $B_1$ be the set of edges in one cone and~$B_2$
    be the set of edges in the other cone. We note that
    $\left| |B_1| - |B_2| \right|$ is equal to one if there is an edge
    $(v,v')\in E$ with $v' \in B_1$ or $v' \in B_2$ and zero otherwise.
    Then, by definition of indegree,
    \begin{align*}
        |\indeg{v}{\dir_1} - &\indeg{v}{\dir_2}| \\
        &= \left| |A|+|B_1|-|A|-|B_2| \right|\\
        &= \left| |B_1|-|B_2| \right|\\
        &= \left| V \cap B \right|,
    \end{align*}
    which holds if and only if $(v,v') \in E$.
    Then $|\indeg{v}{\dir_1} - \indeg{v}{\dir_2}| = 1 \iff (v,v') \in E$,
    as required.
\end{proof}

Next, we prove that we can find the embedding of the edges in the original graph
using $O(n^2)$ directional persistence diagrams.

\begin{theorem}[Edge Reconstruction]
\label{thm:edgeEmbed}
Let $\simComp$ be a plane graph, with vertex set $V$ and edge set $E$.
If~$V$ is known, then we can compute~$E$ using $O(n^2)$ directional
persistence diagrams.
\end{theorem}

\begin{proof}
We prove this theorem constructively. Intuitively, we construct a bow tie for
    each potential edge and use \lemref{edgeExist} to determine if the edge
    exists or not. Our algorithm has three steps for each pair of vertices in
    $V$: Step~$1$ is to
    determine a global bow tie width, Step~$2$ is to construct suitable bow ties,
    and Step~$3$ is to compute indegrees. See \appendref{Example} for an example of
    walking through the reconstruction.

    \emph{Step 1: Determine bow tie width.}
    For each vertex~$v\in V$, we consider the cyclic ordering of the points in
    $V \setminus \{v\}$ around $v$.   We define $\theta(v)$ to be the minimum angle
    between all adjacent pairs of lines through $v$; see \figref{edgebow},
    where the angles between adjacent lines are denoted~$\theta_i$.
    Finally, we choose~$\theta$ less than $\min _{v \in V} \theta(v)$.
    By Lemmas~$1$ and~$2$ of
    \cite{verma2011slow}, we compute the cyclic orderings for all vertices in
    $V$ in~$O(n^2)$ time.
    Since  computing
    each~$\theta(v)$ is~$O(n)$ time once we have the cyclic ordering, the runtime for this step is~$O(n^2)$.

    \emph{Step 2: Constuct bow ties.}
    For each pair of vertices $(v,v') \in V \times V$ such that $v \neq v'$,
    let~$\dir$ be a unit vector perpendicular to vector $(v'-v)$,
    and~let $\dir_1,\dir_2$ be the two unit vectors that form
    angles~$\pm {\theta}$ with~$\dir$. Let $B$ be the bow tie between
     $\pLine{\dir_1 }{\hFiltFun{\dir_1}(v)}$ and~$\pLine{\dir_2}{\hFiltFun{\dir_2}(v)}$.
    Note that by the construction, $B$ contains exactly one point from~$V$, namely~$v'$.

    \begin{figure}
        \center
        \includegraphics{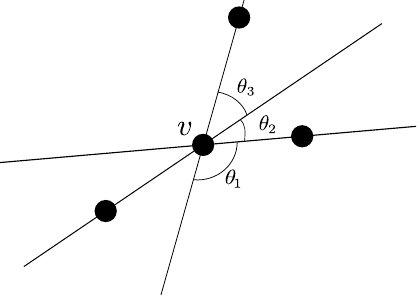}
        \caption{Using a vertex set of a plane graph to construct a bow tie at vertex,
        $v$. Lines are drawn through all vertices and then angles are computed
        between all adjacent pairs of lines. The smallest angle is chosen
        as~$\theta(v)$. Here,~$\theta(v)=\theta_2$.}
        \label{fig:edgebow}
    \end{figure}

    \emph{Step 3: Compute indegrees.}
    Using $B$ as the bow tie in \lemref{Indegree}, compute $\indeg{v}{\dir_1}$ and $\indeg{v}{\dir_2}$.
    Then, using~\lemref{edgeExist}, we determine whether
    $(v,v')$ exists by checking if $|\indeg{v}{\dir_1}-\indeg{v}{\dir_2}|=1$.
    If it does, the edge
    exists; if not, the edge does not.

    Repeating
    for all vertex pairs requires $O(n^2)$
    diagrams and discovers the edges of $\simComp$.
\end{proof}

The implications of \thmref{intComp} and \thmref{edgeEmbed} lead to our primary
result. We can find the embedding of the vertices $V$ by \thmref{intComp}
using three directional persistence diagrams. Furthermore, we can discover edges $E$
with $O(n^2)$ directional persistence diagrams by \thmref{edgeEmbed}. Thus, we can
reconstruct all edges and vertices of a one-dimensional simplicial~complex:

\begin{theorem}[Plane Graph Reconstruction]
Let~$\simComp$ be a plane graph with vertex set $V$ and edge set~$E$.
The vertices, edges, and exact embedding of~$\simComp$ can be determined
using persistence diagrams from~$O(n^2)$ different directions.
\end{theorem}

\section{Discussion}
\label{sec:discussion}

In this paper, we provide an algorithm to  reconstruct a plane graph
with $n$ vertices embedded in~$\R^2$. Our method uses $O(n^{2})$ persistence diagrams
by first determining vertex locations using only three directions, and, second,
determining edge existence based on height filtrations and vertex degrees.
Moreover, if we have an oracle that can return a diagram given a direction in $O(T)$ time,
then constructing the vertices takes $O(T + n \log n)$ and reconstructing the edges takes
takes $O(Tn^2)$ time.

This approach extends to several avenues for future work. First, we plan to generalize
these reconstruction results to higher dimensional simplicial complexes. We can show
that the vertices of a simplicial complex $K$ in~$\R^d$ can be reconstructed in
$O(dT + n^d)$ time using the complete arrangement of hyperplanes and $(d+1)$
directional persistence diagrams.  We conjecture that this bound can be improved to
$O(dT + dn \log n)$ using the same observation that allows us to do the final step of
the vertex reconstruction in linear time.  We have a partial proof in this direction, and
can likewise extend the bow tie idea to higher dimensions, but the number of directions
grows quite quickly. Second, we conjecture that we can reconstruct these plane graphs
with a sub-quadratic number of height filtrations by utilizing more information from each
height filtration. Third, we suspect a similar approach can be used to infer other graph
metrics, such as classifying vertices into connected components. Intuitively, determining
such metrics should require fewer persistence diagrams than required for a complete
reconstruction. Finally, we plan to provide an implementation for reconstruction that
integrates with existing TDA software.

\paragraph{Acknowledgements}
This material is based upon work supported by the National Science Foundation
under the following grants: CCF 1618605 (BTF, SM), DBI 1661530 (BTF, DLM, LW), DGE 1649608 (RLB),
and DMS 1664858 (RLB, BTF, AS, JS).
Additionally, RM thanks the Undergraduate Scholars Program.
All authors thank the CompTaG club at Montana State University and the reviewers
for their thoughtful feedback on this work.

\bibliography{references}

\begin{thebibliography}{10}

\bibitem{maps}
Mahmuda Ahmed, Sophia Karagiorgou, Dieter Pfoser, and Carola Wenk.
\newblock Map construction algorithms.
\newblock In {\em Map Construction Algorithms}, pages 1--14. Springer, 2015.

\bibitem{iti}
Mahmuda Ahmed and Carola Wenk.
\newblock Constructing street networks from {GPS} trajectories.
\newblock In {\em European Symposium on Algorithms}, pages 60--71. Springer,
  2012.

\bibitem{curry2018directions}
Justin Curry, Sayan Mukherjee, and Katharine Turner.
\newblock How many directions determine a shape and other sufficiency results
  for two topological transforms.
\newblock arXiv:1805.09782, 2018.

\bibitem{dey2018graph}
Tamal~K. Dey, Jiayuan Wang, and Yusu Wang.
\newblock Graph reconstruction by discrete {M}orse theory.
\newblock arXiv:1803.05093, 2018.

\bibitem{edelsbrunner2010computational}
Herbert Edelsbrunner and John Harer.
\newblock {\em Computational Topology: {A}n Introduction}.
\newblock American Mathematical Society, 2010.

\bibitem{ge2011data}
Xiaoyin Ge, Issam~I Safa, Mikhail Belkin, and Yusu Wang.
\newblock Data skeletonization via {R}eeb graphs.
\newblock In {\em Advances in Neural Information Processing Systems}, pages
  837--845, 2011.

\bibitem{ghrist2018euler}
Robert Ghrist, Rachel Levanger, and Huy Mai.
\newblock Persistent homology and {E}uler integral transforms.
\newblock arXiv:1804.04740, 2018.

\bibitem{giusti2015clique}
Chad Giusti, Eva Pastalkova, Carina Curto, and Vladimir Itskov.
\newblock Clique topology reveals intrinsic geometric structure in neural
  correlations.
\newblock {\em Proceedings of the National Academy of Sciences},
  112(44):13455--13460, 2015.

\bibitem{ili}
Sophia Karagiorgou and Dieter Pfoser.
\newblock On vehicle tracking data-based road network generation.
\newblock In {\em SIGSPATIAL '12: Proceedings of the 20th International
  Conference on Advances in Geographic Information Systems}, pages 89--98. ACM,
  2012.

\bibitem{kegl2000learning}
Bal{\'a}zs K{\'e}gl, Adam Krzyzak, Tam{\'a}s Linder, and Kenneth Zeger.
\newblock Learning and design of principal curves.
\newblock {\em IEEE Transactions on Pattern Analysis and Machine Intelligence},
  22(3):281--297, 2000.

\bibitem{lee2017quantifying}
Yongjin Lee, Senja~D. Barthel, Pawe{\l} D{\l}otko, S.~Mohamad Moosavi, Kathryn
  Hess, and Berend Smit.
\newblock Quantifying similarity of pore-geometry in nanoporous materials.
\newblock {\em Nature Communications}, 8:15396, 2017.

\bibitem{verma2011slow}
David~L. Millman and Vishal Verma.
\newblock A slow algorithm for computing the {G}abriel graph with double
  precision.
\newblock {\em CCCG '11: Proceedings of the 23rd Annual Canadian Conference on
  Computational Geometry}, 2011.

\bibitem{nicolau2011topology}
Monica Nicolau, Arnold~J. Levine, and Gunnar Carlsson.
\newblock Topology based data analysis identifies a subgroup of breast cancers
  with a unique mutational profile and excellent survival.
\newblock {\em Proceedings of the National Academy of Sciences},
  108(17):7265--7270, 2011.

\bibitem{rizvi2017single}
Abbas~H. Rizvi, Pablo~G. Camara, Elena~K. Kandror, Thomas~J. Roberts, Ira
  Schieren, Tom Maniatis, and Raul Rabadan.
\newblock Single-cell topological {RNA}-seq analysis reveals insights into
  cellular differentiation and development.
\newblock {\em Nature Biotechnology}, 35(6):551, 2017.

\bibitem{turner2014persistent}
Katharine Turner, Sayan Mukherjee, and Doug~M. Boyer.
\newblock Persistent homology transform for modeling shapes and surfaces.
\newblock {\em Information and Inference: A Journal of the IMA}, 3(4):310--344,
  2014.

\bibitem{rootreconstruction}
Ying Zheng, Steve Gu, Herbert Edelsbrunner, Carlo Tomasi, and Philip Benfey.
\newblock Detailed reconstruction of 3d plant root shape.
\newblock {\em Proceedings of the IEEE International Conference on Computer
  Vision}, pages 2026--2033, 11 2011.

\end{thebibliography}

\section*{Appendix}
\appendix

\section{Example of Reconstructing a Plane Graph}
\label{sec:Example}

We give an example of reconstructing a plane graph.
Consider the complex given in \figref{exampleVertex}.

\paragraph{Vertex Reconstruction}

First, we find vertex locations using the algorithm described in \secref{vRec}.
We need to choose pairwise linearly independent
vectors~$\dir_1, \dir_2$ and~$\dir_3$ such that only~$n$ three-way intersections in
$A = \pLines{\dir_1}{V} \cup \pLines{\dir_2}{V} \cup \pLines{\dir_3}{V}$
exist; note that in this example,~$n=4$.  Using the persistence diagrams from
height filtrations in directions $\dir_1=(0,1)$ and $\dir_2=(1,0)$, we construct the set
of lines~$\pLines{\dir_1}{V} \cup \pLines{\dir_2}{V}$. This results in~$n^2 = 16$ possible
locations for the vertices at the intersections in~$A$. We show these filtration lines and
intersections in \figref{Vertx-second_dir}. Next, we compute the third direction~$\dir_3$
using the algorithm outlined in \thmref{intComp}. To do this, we need to find the greatest
horizontal distance between two vertical lines,~$d_1 = 2$ and the least vertical distance
between two horizontal lines,~$d_2=1$. Then, we use these to choose a direction~$\dir_3$
perpendicular to~$\dir_{*} = (d_1, \frac{d_2}{2}) = (2, \frac{1}{2})$
(e.g.,~$\dir_3 = (\frac{-1}{\sqrt{17}},
\frac{4}{\sqrt{17}}) \in \sph^1$). Then, the four three-way intersections
 in~$\pLines{\dir_1}{V} \cup \pLines{\dir_2}{V} \cup \pLines{\dir_3}{V}$
identify all Cartesian coordinates of the original complex. We show filtration lines from
all three directions in \figref{Vertex-third_dir}.

\paragraph{Edge Reconstruction}
Next, we reconstruct all edges as described in \secref{eRec}. In order to do so, we first
find the~$\theta$ we will use to construct bow ties. To do this, we examine each
vertex~$v$ in turn, finding~$\theta(v)$, the minimum angle between adjacent pairs of
lines through~$v$ and~$v' \in V - \{v \}$. Ordering~$v$ by increasing $x$-coordinate,
we find~$\theta(v)$ to be approximately ~$0.237, 0.219, 0.399$, and~$0.180$ radians, respectively. Then, we
take~$\theta$ to be less than the minimum of these, i.e.~$<0.180 radians$.

Now, for each of the~$\frac{n(n-1)}{2}$ pairs of vertices~$(v,v') \in V^2$, we construct
a bow tie~$B$ and then use this bow tie to determine whether an edge exists between
the two vertices. We go through two examples: one for a pair of vertices that does have an
edge between, and one for a pair that does not. First, consider the pair~$v = (0.25,0)$
and~$v' = (1,1)$. To construct their bow tie, we first find the unit vector perpendicular
to the vector that points from~$v$ to~$v'$, which is~$\dir = (-0.8,0.6)$. Now, we
find~$\dir_1,\dir_2$ such that they make angles~$\theta$ with~$\dir$.
We choose~$\dir_1 = (-0.956, 0.293)$ and~$\dir_2 = (-0.433,0.902)$. Now, by
\lemref{Indegree}, we can use the persistence diagrams from these two directions to
compute~$\indeg{v}{\dir_1}$ and~$\indeg{v}{\dir_2}$. We observe that $\dgm{0}{\dir_1}$
contains exactly one birth-death pair~$(x,y)$ such that~$y=v \cdot \dir_1$ and
$\dgm{1}{\dir_1}$ has one birth-death pair such
that~$x=v \cdot \dir_1$. Thus,~$\indeg{v}{\dir_1}=2$. On the other hand, $\dgm{0}{\dir_2}$
contains exactly one birth-death pair~$(x,y)$ such that~$y=v \cdot \dir_2$,
but~$\dgm{1}{\dir_2}$ contains no birth-death pair such that~$x=v \cdot \dir_2$.
So~$\indeg{v}{\dir_2}=1$. Now, since~$|\indeg{v}{\dir_1}-\indeg{v}{\dir_2}|=1$, we know
that~$(v,v') \in E$, by \lemref{edgeExist}.

For the second example, consider the pair of vertices~$v = (0.25,0)$ and~$v' = (-1,2)$.
Again, we construct their bow tie by finding a unit vector perpendicular to the vector
pointing from~$v$ to~$v'$. We choose this~$s=(0.848,0.530)$. Then, the~$\dir_1$
and~$\dir_2$ which form angle~$\theta<0.180$ radians (e.g $\theta=.170$) with~$s$
are~$ \dir_1 = (0.968, 0.248)$ and~$\dir_2 = (0.472, 0.882)$. Again by \lemref{Indegree},
we examine the zero- and one-dimensional persistence diagrams from these two directions
to compute the indegree from each direction for vertex~$v$. In $\dgm{0}{\dir_1}$, we
have one pair~$(x,y)$ which dies at~$y=v \cdot \dir_1$, but in~$\dgm{1}{\dir_1}$, no pair
is born at~$x=v \cdot \dir_1$. So~$\indeg{v}{\dir_1}=1$. We see the exact same
for~$\dir_2$, which means that~$|\indeg{v}{\dir_1}-\indeg{v}{\dir_2}|=0$. Since
\lemref{edgeExist} tells us that we have an edge between~$v$ and~$v'$ only if the
absolute value of the difference of indegrees is one, we know that there is no edge between
vertices~$(0.25,0)$ and~$(-1,2)$.

In order to reconstruct all edges, we perform the same computations for all pairs
of vertices.

\todo{update figures so that the (0,1) and (1,0) are swapped (follow what's going on in the paper)}
\begin{figure*}
\centering
\begin{subfigure}{.3\linewidth}
  \centering
  \includegraphics[width=\linewidth,height=.8\linewidth]{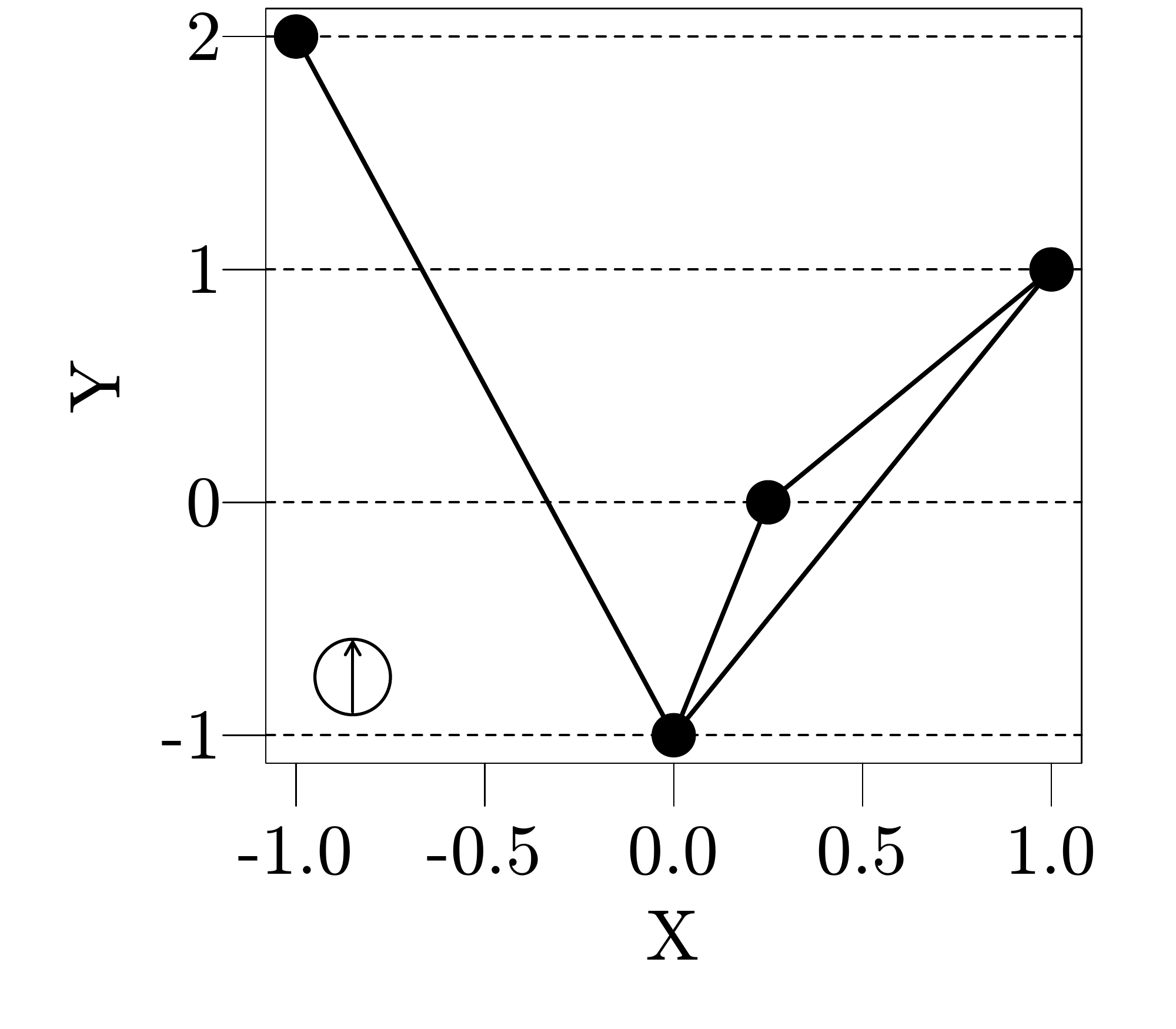}
  \caption{Filtration lines for $\dir_1$}
\end{subfigure}%
\begin{subfigure}{.3\linewidth}
  \centering
  \includegraphics[width=\linewidth,height=.8\linewidth]{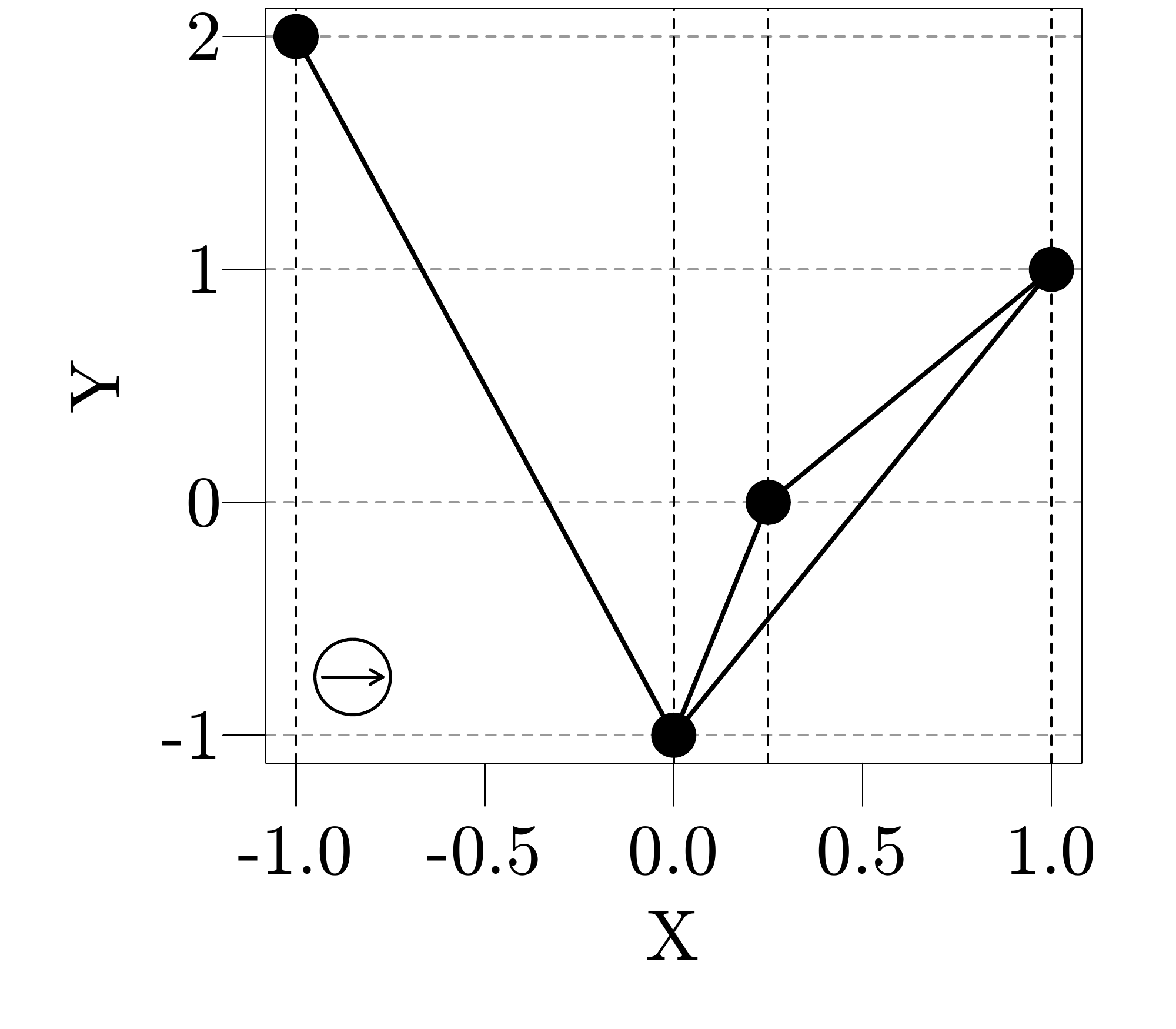}
  \caption{Filtration lines for $\dir_2$}
  \label{fig:Vertx-second_dir}
\end{subfigure}%
\begin{subfigure}{.3\linewidth}
  \centering
  \includegraphics[width=\linewidth,height=.8\linewidth]{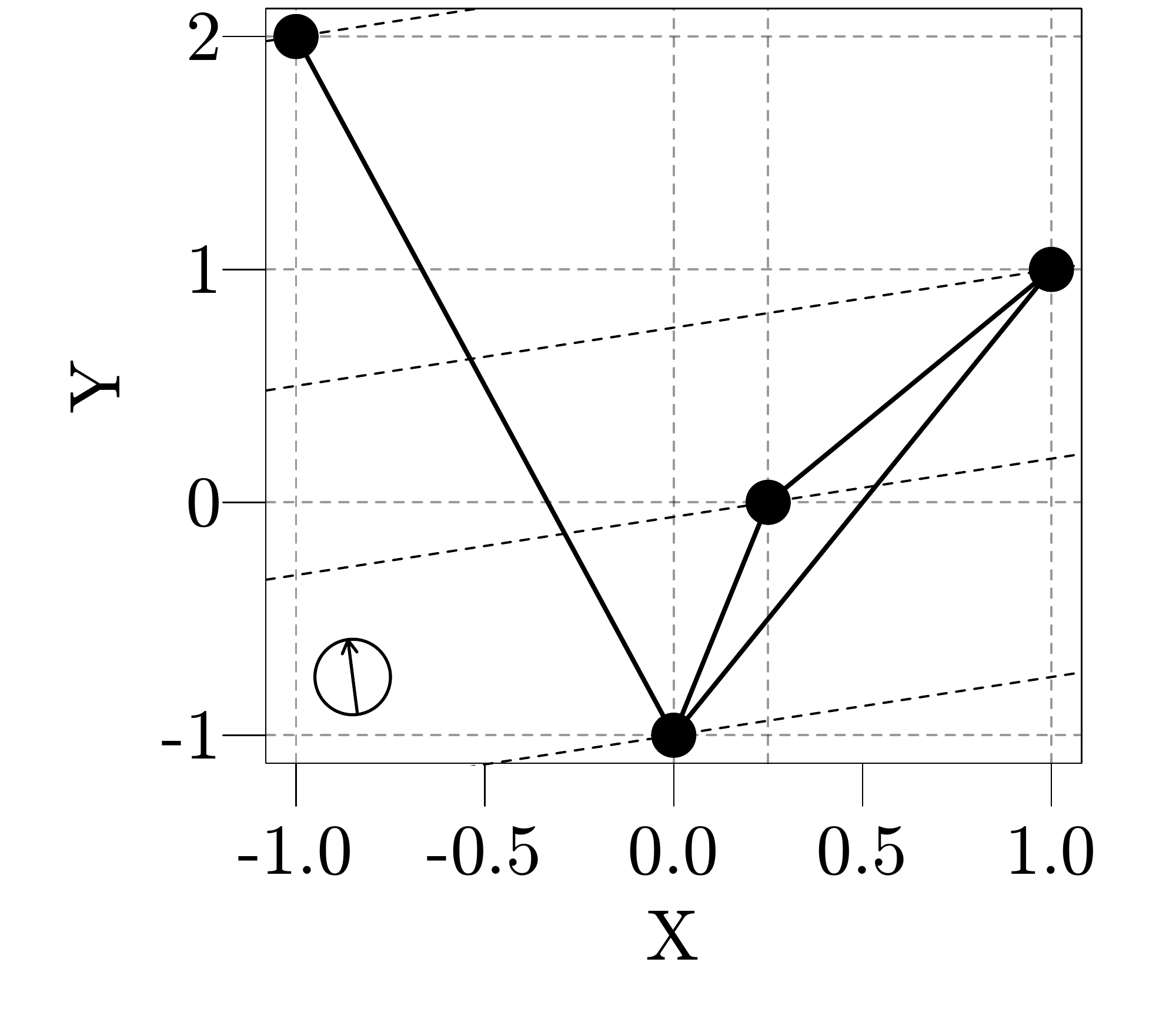}
  \caption{Filtration lines for $\dir_3$}
    \label{fig:Vertex-third_dir}
\end{subfigure}\\
\vspace{1cm}
\begin{subfigure}{.3\linewidth}
  \centering
  \includegraphics[width=\linewidth,height=.8\linewidth]{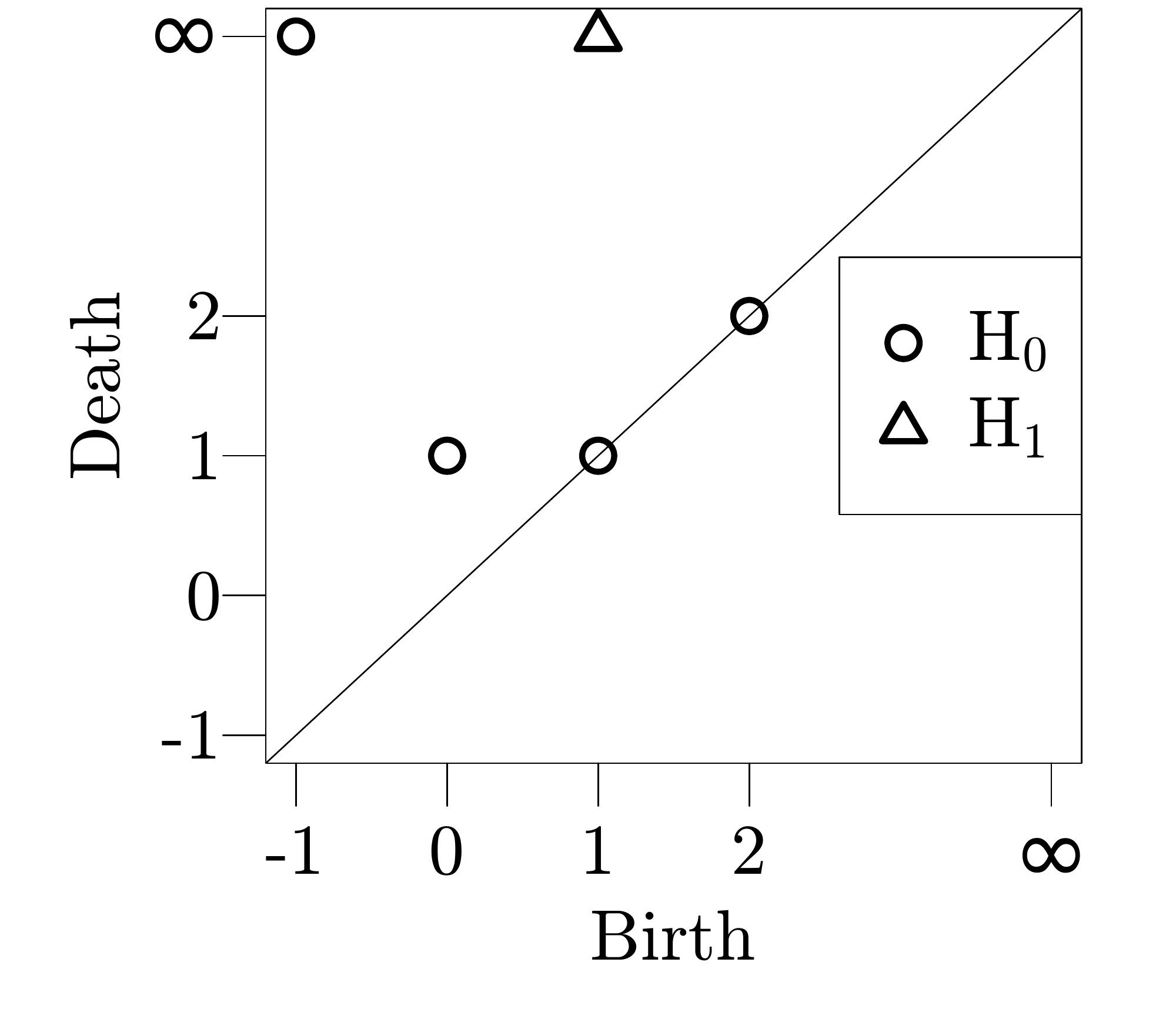}
  \caption{Diagrams for $\dir_1$}
\end{subfigure}
\begin{subfigure}{.3\linewidth}
  \centering
  \includegraphics[width=\linewidth,height=.8\linewidth]{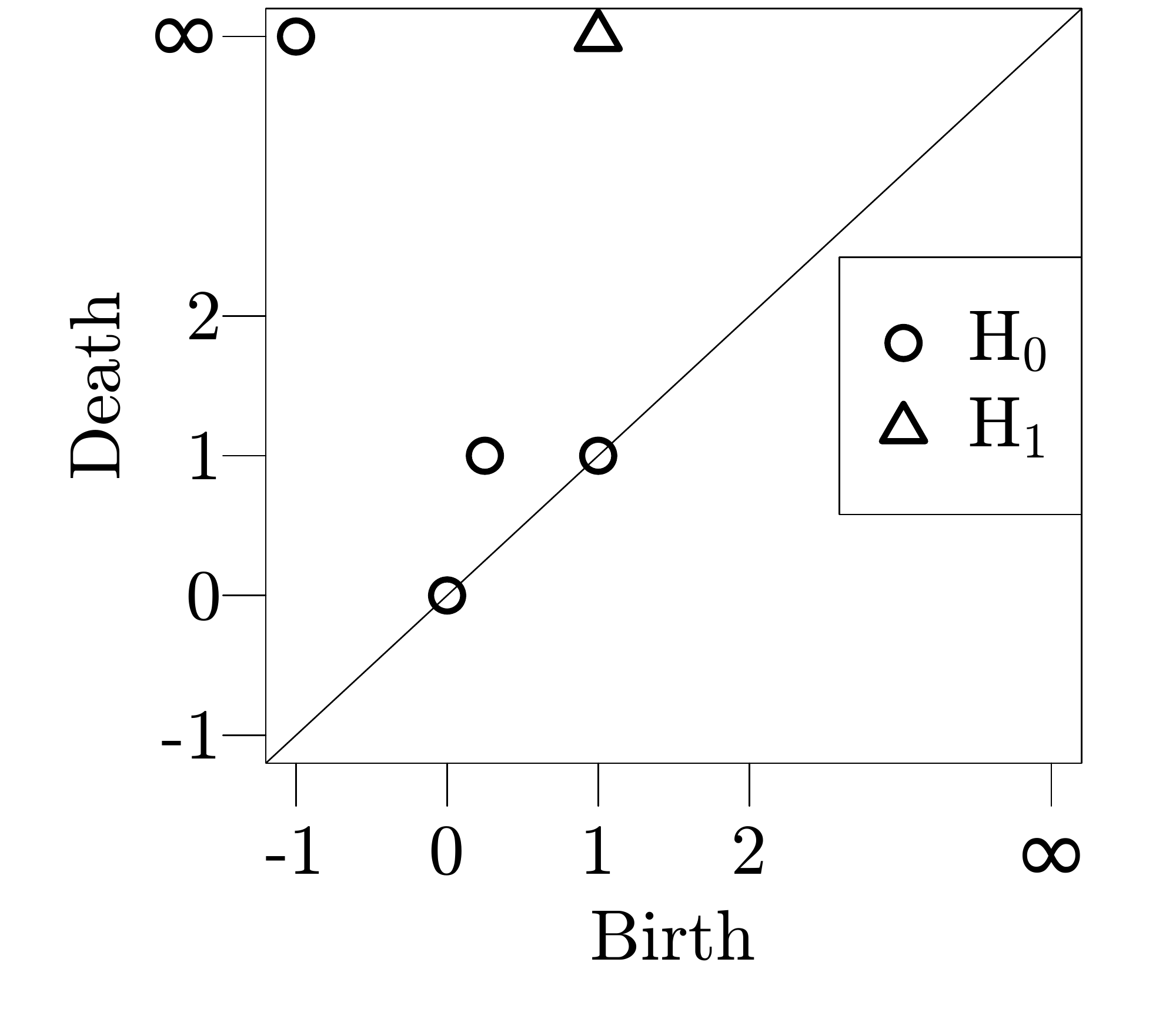}
  \caption{Diagrams for $\dir_2$}
\end{subfigure}%
\begin{subfigure}{.3\linewidth}
  \centering
  \includegraphics[width=\linewidth,height=.8\linewidth]{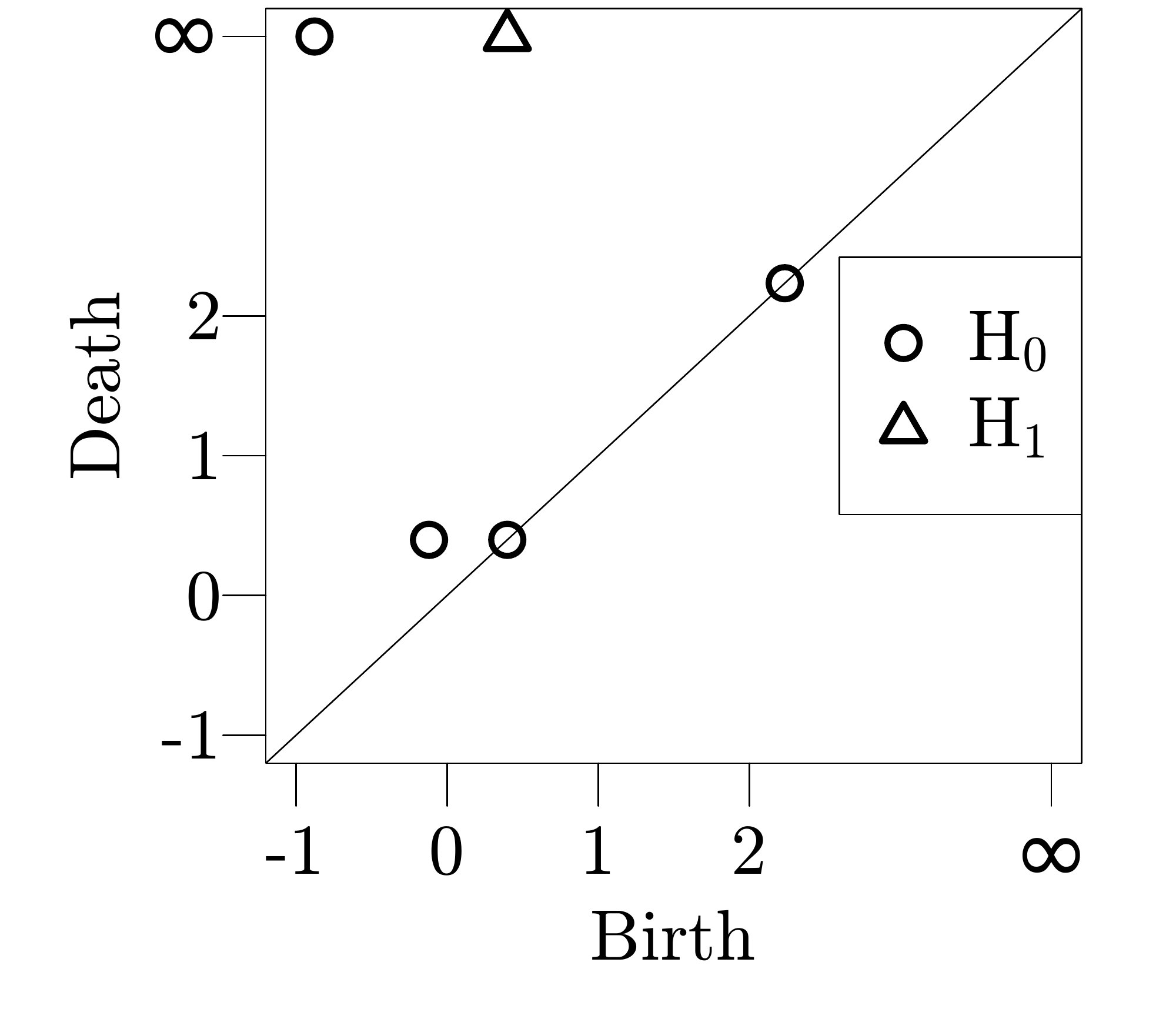}
  \caption{Diagrams for $\dir_3$}
\end{subfigure}
  \caption{Example of vertex reconstruction from three directions, $\dir_1$,
  $\dir_2$ and $\dir_3$ with corresponding persistence diagrams built for
  height filtrations from these directions. The filtration lines are the dotted
  lines superimposed over the complex.}
  \label{fig:exampleVertex}
\end{figure*}

\begin{figure*}
\centering
\begin{subfigure}{.3\linewidth}
  \centering
  \includegraphics[width=\linewidth,height=.8\linewidth]{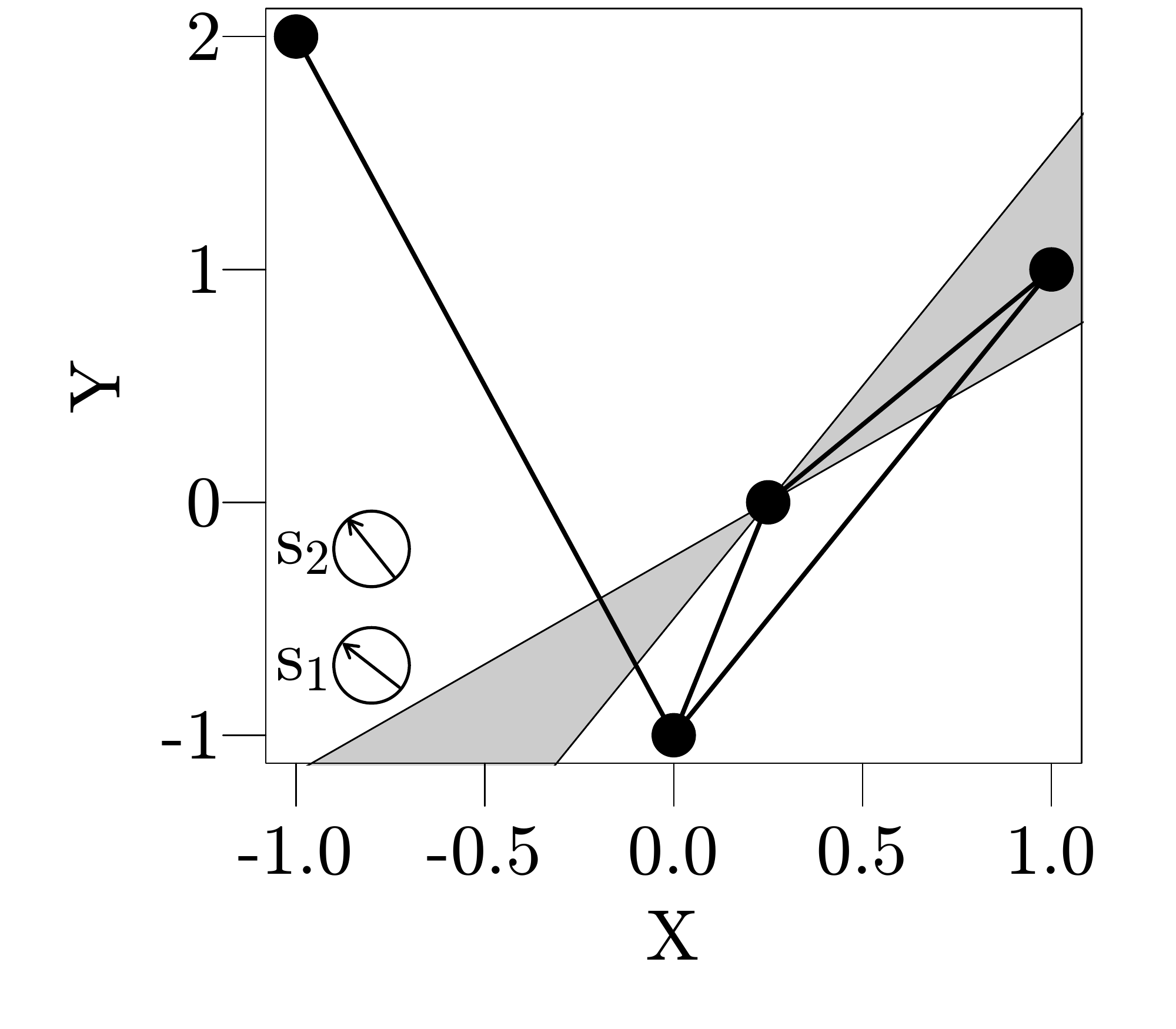}
  \caption{Bow tie lines for $\dir_1$ and $\dir_2$}
\end{subfigure}%
\begin{subfigure}{.3\linewidth}
  \centering
  \includegraphics[width=\linewidth,height=.8\linewidth]{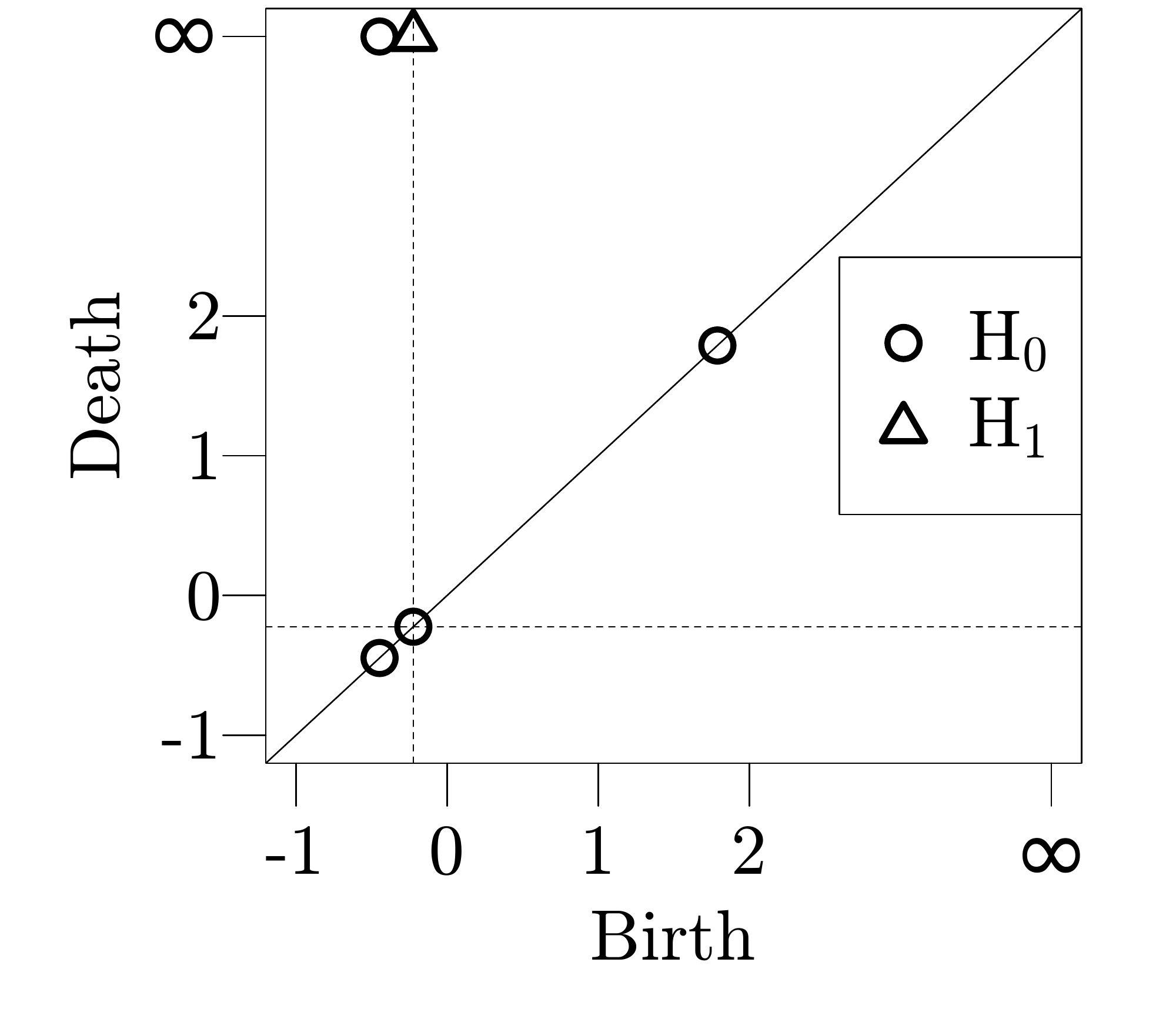}
  \caption{Diagram for $\dir_1$}
  \label{fig:Edge-second_dir}
\end{subfigure}%
\begin{subfigure}{.3\linewidth}
  \centering
  \includegraphics[width=\linewidth,height=.8\linewidth]{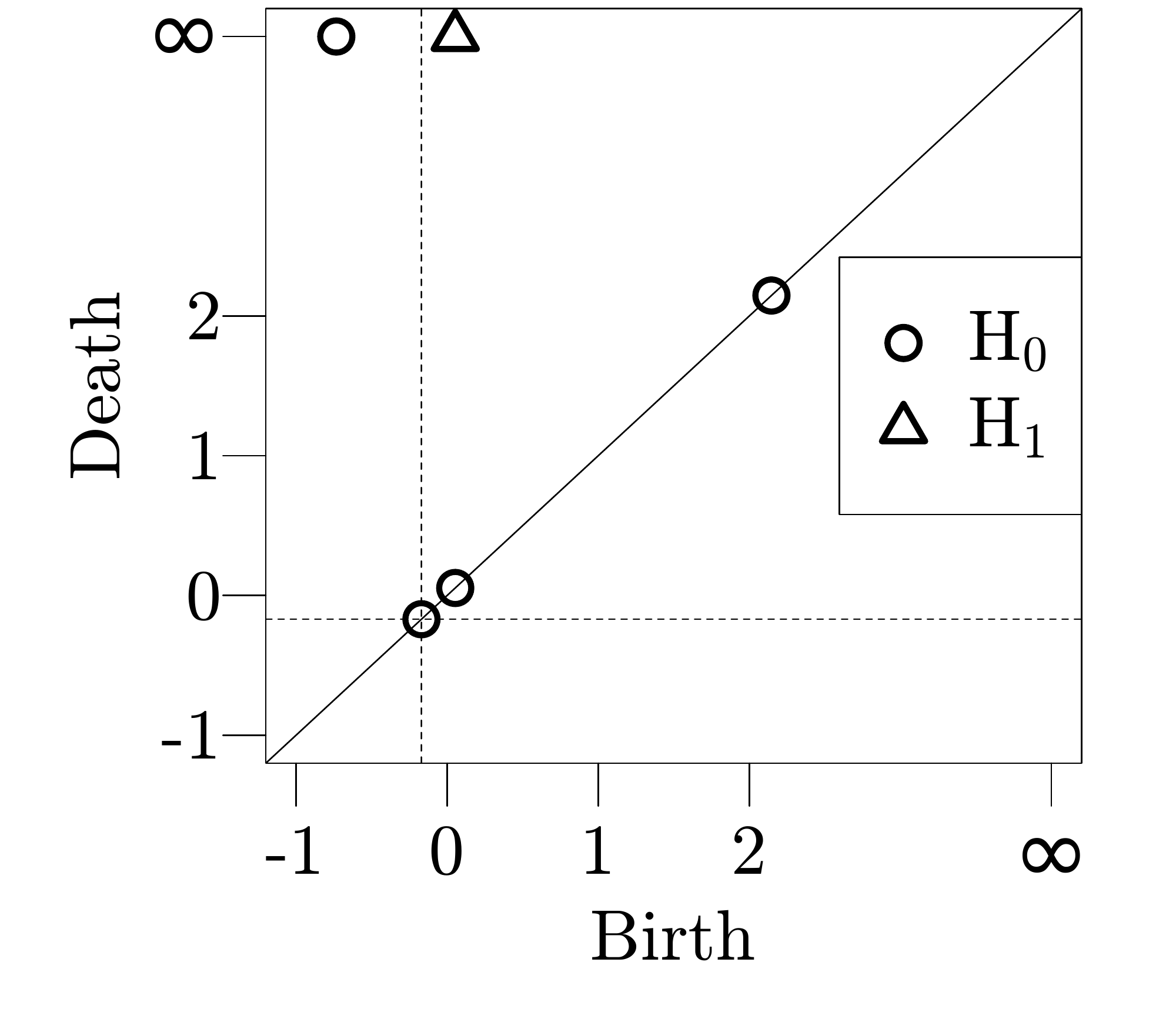}
  \caption{Diagram for $\dir_2$}
    \label{fig:Edge-third_dir}
\end{subfigure}\\
\vspace{1cm}
\begin{subfigure}{.3\linewidth}
  \centering
  \includegraphics[width=\linewidth,height=.8\linewidth]{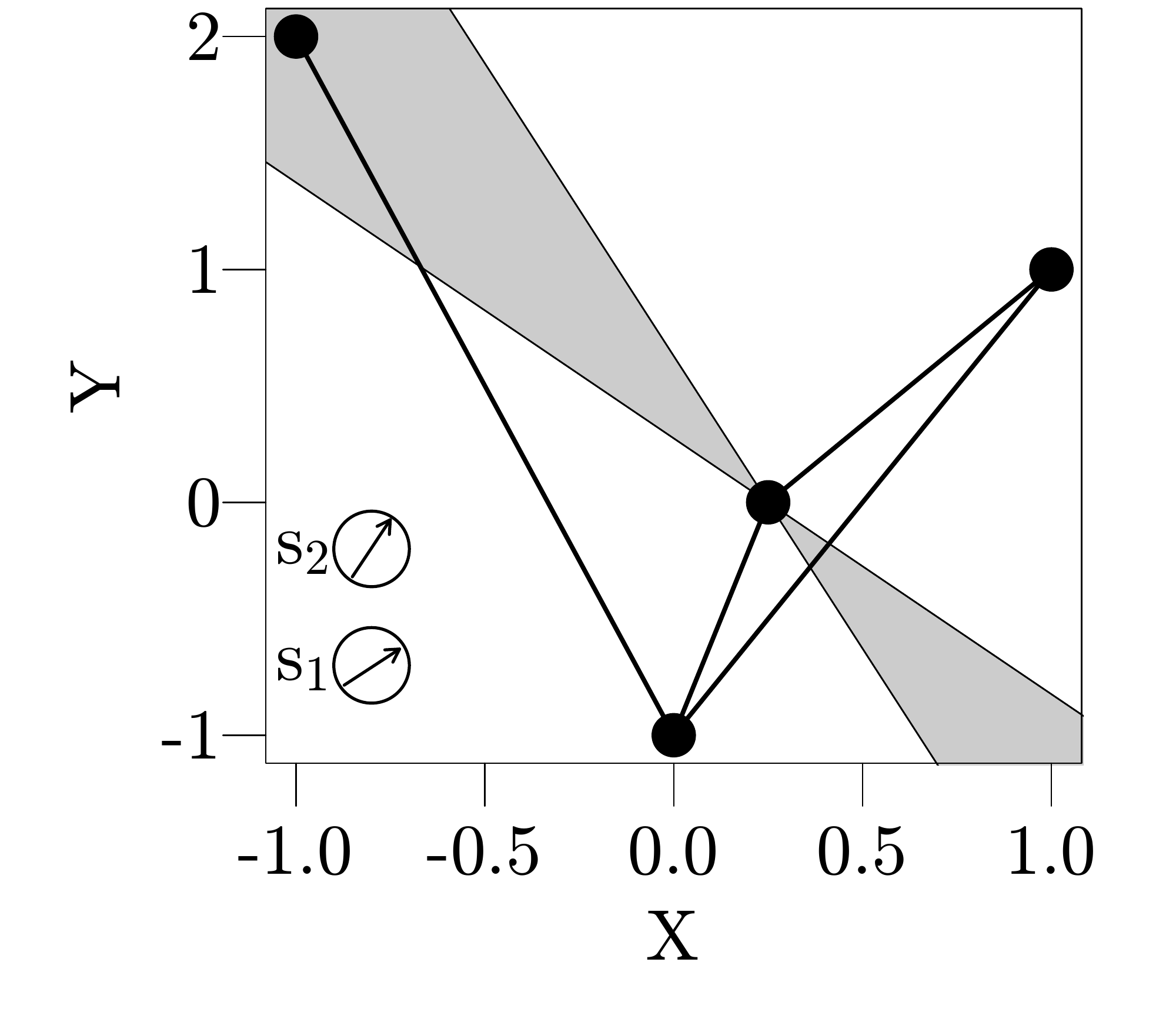}
  \caption{Bow tie lines for $\dir_1$ and $\dir_2$}
\end{subfigure}
\begin{subfigure}{.3\linewidth}
  \centering
  \includegraphics[width=\linewidth,height=.8\linewidth]{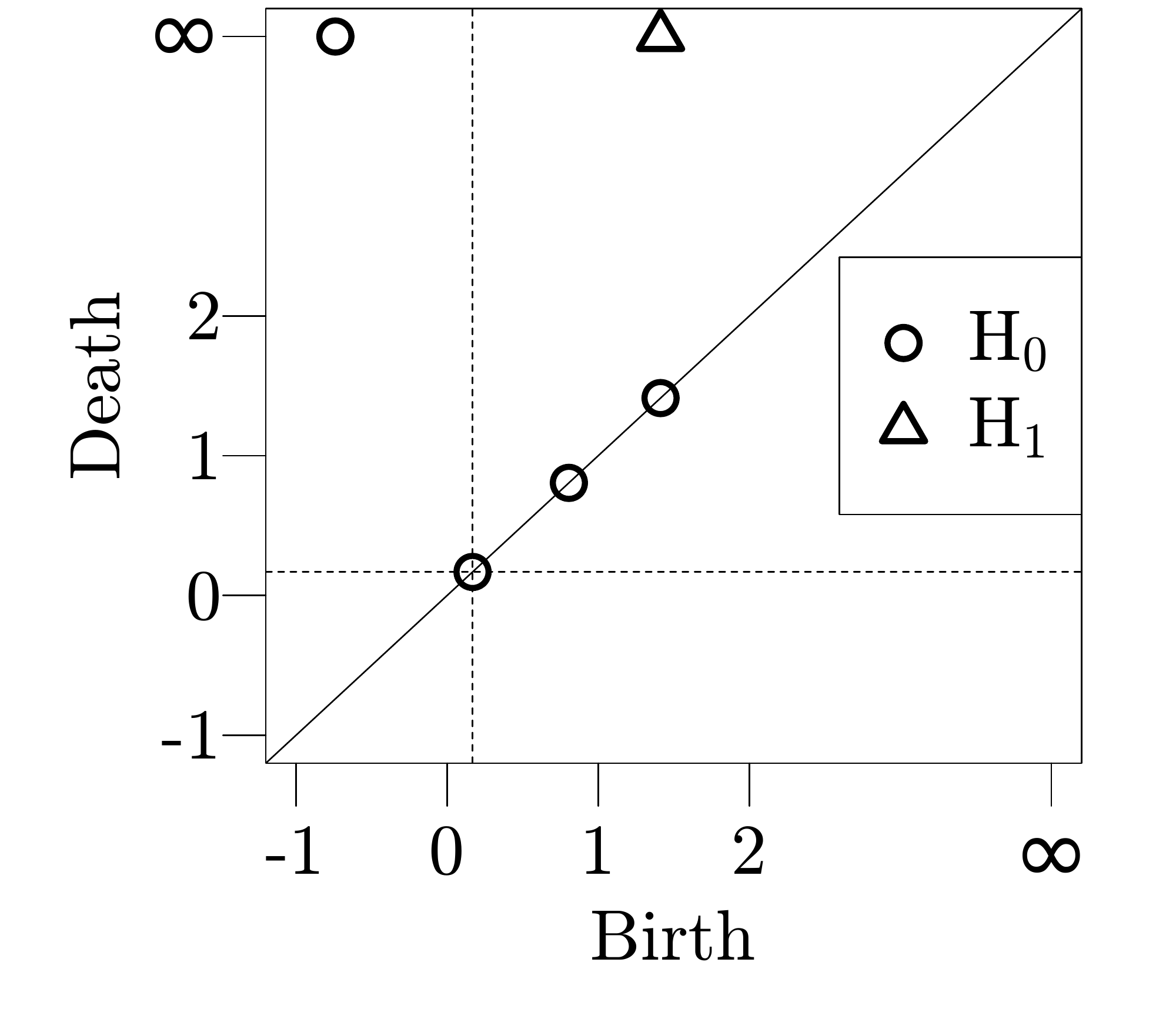}
  \caption{Diagram for $\dir_1$}
\end{subfigure}%
\begin{subfigure}{.3\linewidth}
  \centering
  \includegraphics[width=\linewidth,height=.8\linewidth]{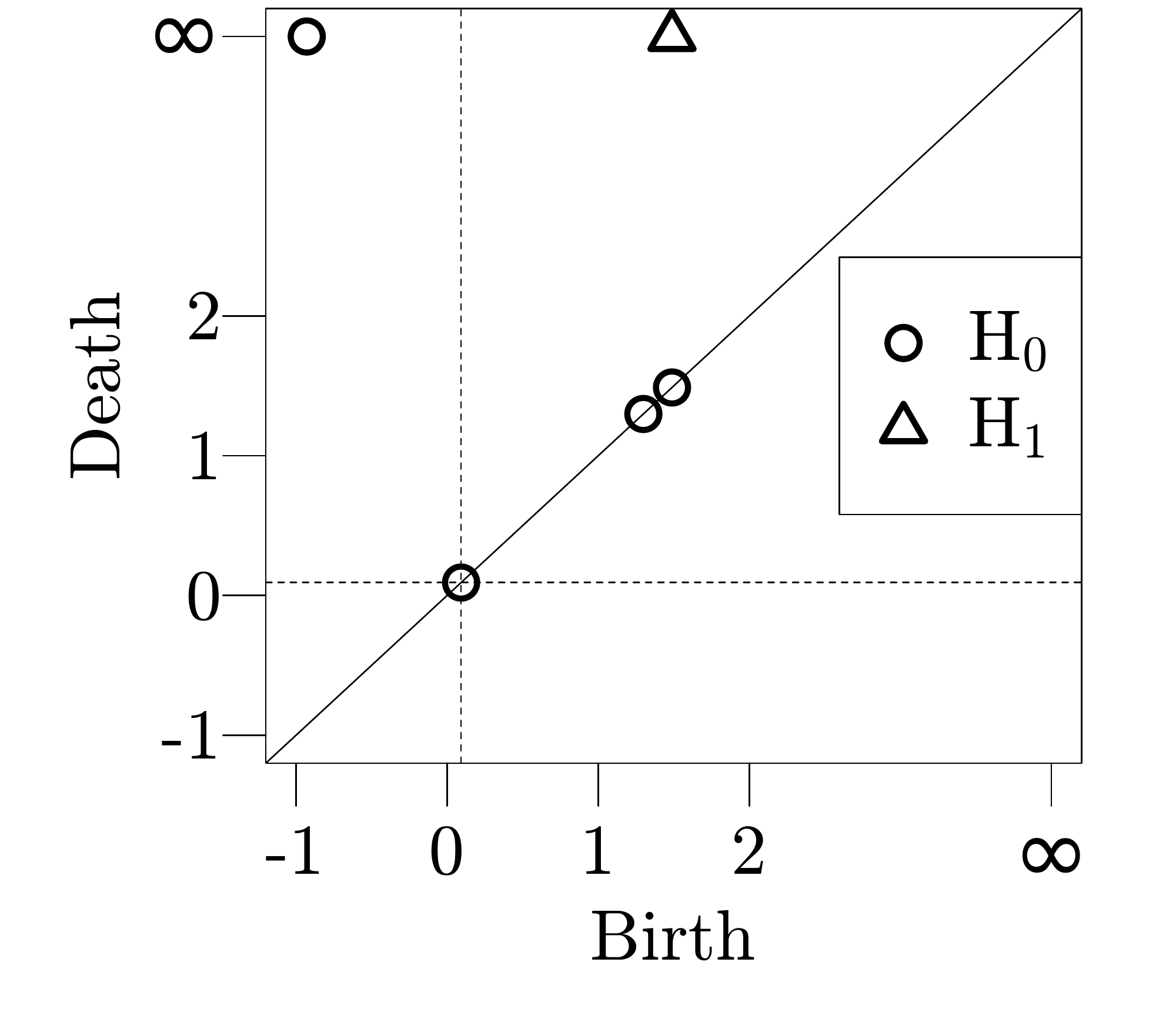}
  \caption{Diagram for $\dir_2$}
\end{subfigure}
  \caption{Example of edge reconstruction for two edges. The first edge (top
  row) exists while the second edge (bottom row) does not. The bow tie
  is given on the left while the persistence diagrams $\dgm{0}{\dir_1}$ and
  $\dgm{1}{\dir_1}$ are given in the middle and the persistence diagrams
  $\dgm{0}{\dir_2}$ and $\dgm{1}{\dir_2}$ are given on the right. The dotted lines
  indicate $v \cdot \dir_1$ and $v \cdot \dir_2$ in diagrams for $\dir_1$ and $\dir_2$ respectively.}
  \label{fig:exampleEdge}
\end{figure*}

\end{document}